\documentclass[runningheads]{llncs}

\usepackage{graphicx}
\usepackage[utf8]{inputenc}
\usepackage{xcolor}

\usepackage{amsbsy}
\usepackage{amsmath}
\usepackage{amsfonts}
\usepackage{amssymb}
\usepackage{graphicx}
\usepackage{csquotes}
\usepackage{enumitem}
\usepackage{booktabs,array}
\usepackage[frozencache,cachedir=.]{minted}

\usepackage{hyperref}

\graphicspath{{figs/}}

\spnewtheorem{observation}{Observation}{\bfseries}{\itshape}

\begin{document}
\title{Side-Contact Representations with Convex Polygons in 3D: New Results for Complete Bipartite Graphs }
\titlerunning{Side-contact representations in 3D for complete bipartite graphs.}
%
\author{
Andr\'e Schulz\orcidID{0000-0002-2134-4852} }

\authorrunning{A. Schulz}

%
\institute{FernUniversität in Hagen, Universitätsstraße 47
58097 Hagen, Germany 
\email{andre.schulz@fernuni-hagen.de}}
\maketitle              
\begin{abstract}
A  polyhedral surface~$\mathcal{C}$ in $\mathbb{R}^3$ with convex polygons as faces 
is a \emph{side-contact representation} of a graph~$G$ if there is a bijection
between the vertices of $G$ and the faces of~$\mathcal{C}$ such that the polygons
of adjacent vertices are exactly the polygons sharing an entire  common side in~$\mathcal{C}$.

We show that $K_{3,8}$ has a side-contact representation but 
$K_{3,250}$ has not. The latter result implies that the number of edges of a 
graph with side-contact representation and $n$ vertices is bounded by $O(n^{5/3})$. 

\end{abstract}

\section{Introduction}

Contact representations are a classical approach to visualize graphs.
A graph $G$ has a contact representation
if there is a bijection between its vertex set $V$ and a set 
of interior-disjoint geometric objects from
a given class such that two objects touch if and only if the corresponding
vertices are adjacent.  
For a concrete contact representation it has to be 
specified, which geometric objects are considered 
(including their embedding space)
and what it means for two objects to touch. In this paper 
we consider convex polygons in 3D as geometric objects. 
To avoid confusion we call the edges
of a polygon its \emph{sides} and the vertices its \emph{corners}. Two polygons touch with a \emph{side-contact} if and only if
they have a full side in common.
It is not allowed that a side is contained in more than two polygons.
 Notice that we do not require that all polygon sides
are incident to two polygons. For brevity, we call  
 representations of convex polygons in 3D
with such side-contacts simply \emph{side-contact representations} throughout the paper
and every polygon will be considered as convex.

It is an open question to characterize the graphs
that have a side-contact representation. First results were given by Arseneva
et al.~\cite{akklsvw}, who introduced this kind of contact representation.
We list some of the results from Arseneva
et al.: Exactly the planar graphs have a side-contact representation in the plane.
The graph $K_5$ has no side-contact representation in 3D, but $K_{3,5}$ and $K_{4,4}$ have one. Another
graph that has no side-contact representation is $K_{5,81}$, which implies by the K\H{o}vari--S\'os--Tur\'an
  theorem~\cite{KST54} that graphs with side-contact representation have at most $O(n^{9/5})$ edges, for
  $n$ being the number of vertices. On the other hand all graphs of hypercubes have a side-contact representation
  and thus there are $n$-vertex graphs with $\Theta(n\log n)$ edges with side-contact representation. 

There exists a large body of literature for other types 
of contact representations. For a few selected results in 2D we redirect 
the reader to Arseneva et al.~\cite{akklsvw}. For 3D we list some selected results here:
Due to Tietze~\cite{tietze1905} every graph has a contact representation with
 interior-disjoint convex polytopes in $\mathbb{R}^3$ where contacts are given by  shared 2-dimensional
facets. 
Evans et al. showed that every graph has a
contact representation in 3D where two convex polygons touch if they share a single corner~\cite{erssw-rghtp-GD19}.
Every
planar graph has a contact
representation with axis-parallel cubes in~$\mathbb{R}^3$
as shown by Felsner and Francis~\cite{FF11} (two cubes touch if their boundaries intersect), 
a similar result for boxes was discovered earlier by 
Thomassen~\cite{t86}.
Kleist and Rahman~\cite{kleist14} studied a similar model but required that 
the intersection has nonzero area. They proved
that every subgraph of an Archimedean
grid 
can be represented with unit cubes. Alam et al.~\cite{alam15} showed in the same model
with axis-aligned boxes that every 3-connected planar graph and its dual can be 
represented simultaneously. 

\paragraph{Our contribution.}  We extend the results of Arseneva
et al.~\cite{akklsvw} for complete bipartite graphs. We
 construct  a side-contact representation of $K_{3,8}$, where 
previously only a construction for $K_{3,5}$ was known. On the other
hand, we prove that $K_{3,250}$ has no side-contact representation. As a consequence
the number of edges of an $n$-vertex graph with a side-contact representation 
is bounded by $O(n^{5/3})$.

\section{A side-contact representation for $K_{3,8}$}
\label{sec:lb}

In this section we explain how to construct a side-contact representation of $K_{3,8}$.
As an intermediate step we
construct a \emph{corner-contact representation}. In contrast
to side-contact representations, two polygons touch with corner-contact if they share a single corner.
For a polygon $p$ its supporting plane $p^=$ defines two open half-spaces which we label
arbitrarily as $p^+$ and $p^-$.
We say that a representation (either corner-contact or side-contact) is \emph{one-sided}, if for every   
polygon $p$ its touching polygons lie either all in the closure of $p^+$, or they lie all in the 
closure of $p^-$. Note that the definition of one-sided is slightly stronger than the definition of 
\emph{one-sided with respect to a set} as used by by Arseneva
et al.~\cite{akklsvw}.

\begin{lemma}\label{lem:transform}
 Every one-sided corner-contact representation of $K_{3,8}$ can be transformed
 into a side-contact representation of $K_{3,8}$.
\end{lemma}
\begin{proof}
We call the polygons from the partition class with eight elements the blue polygons.
Every blue polygon $b$ can be trimmed to a triangle touching the (red) polygons
$r_1,r_2,r_3$. We can assume for every $b$ that the red polygons lie in $b^+\cup b^=$.
Consider the plane $b^=$ and the line arrangement $\mathcal{A}$ 
given by $b^=\cap r^=_i$ for $i\in\{1,2,3\}$. Since the representation is one-sided,
$\mathcal{A}$ contains a triangular cell $\Delta$ such that every edge of $\Delta$ contains exactly one
corner of $b$. Let $h$ be a plane parallel to $b^=$  (inside $b^+$, very close to $b^=$) such that the
line arrangement 
given by $h\cap r^=_i$ for $i\in\{1,2,3\}$ is combinatorially equivalent to $\mathcal{A}$
and furthermore the cell corresponding to $\Delta$ contains on every edge exactly one 
line segment from $S=\{s_i:=h\cap r_i\mid i\in \{1,2,3\}\}$. We can now replace $b$ by the convex hull of $S$ and then restrict the red polygons 
to $b^=\cup b^+$, w.r.t. the modified $b$. 
Since the three segments of $S$ lie on the boundary of $\Delta$, they all appear on 
its convex hull.
Thus we keep all incidences
without introducing new ones (see~\autoref{fig:transform}). Also, the one-sidedness property is maintained.
Repeating this for every blue polygon
yields a side-contact representation of $K_{3,8}$. 
\end{proof}
\begin{figure}
	\centering
	\includegraphics[scale=.77,page=2]{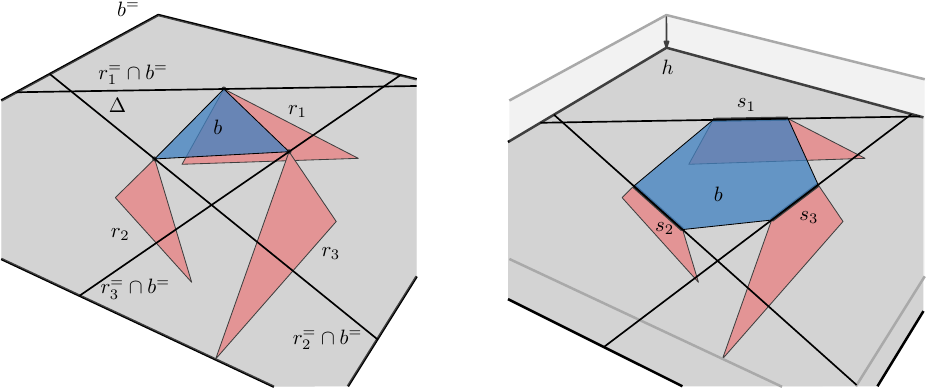}
	\caption{Offsetting the supporting planes of the blue polygons can transform a one-sided corner-contact representation
	into a side-contact representation.}
	\label{fig:transform}
\end{figure}

It remains to construct a one-sided corner-contact representation for $K_{3,8}$.
We start with a hexagonal prism of height $1$. 
Its base is parallel to the xy-plane and given by a hexagon with alternating side lengths
2 and 14 and interior angles of $2\pi/3$. We name the corners of the bottom
base (in cyclic order) $x_0,\ldots,x_5$, and the corners at the top
$x'_0,\ldots,x'_5$, such that $x_i$ and $x'_i$ are adjacent.
All indices of these points are considered modulo 6.
Let $\ell_i$ be the segment
between $x_{2i+1}$ and $x'_{2i+4}$ for $i\in\{1,2,3\}$. For any $\ell_i$ we 
define $\ell'_i$ to be a copy of $\ell_i$ that is vertically shifted up by $0.2$.
Now, we subdivide all six segments in the middle and move the subdivision point vertically up by $1.08$
in case of the $\ell'_i$s and vertically down by $1.08$ in case of the $\ell_i$s.
We define for all $i\in\{1,2,3\}$ the (red) polygon $r_i$ as 
the convex hull of $\ell_i$ and $\ell'_i$ (including the translated subdivision point). 
Notice that the polygons are disjoint (see appendix).
The convex hull of these polygons defines a convex polyhedron $\mathcal{P}$.
We observe that  $\mathcal{P}$ has eight triangular faces that are incident to all
red polygons. These define the blue polygons (see~\autoref{abb:k38}). Note that at the subdivision points we have 
one red polygon adjacent to two blue polygons. To resolve this issue we replace all subdivision
points by an $\varepsilon$-small side such that the red polygons remain convex
and all blue polygons still appear on the convex hull of $\{r_1,r_2,r_3\}$ (details are given in the appendix).
We then move the corners of the adjacent blue polygons to two distinct endpoints of the new
sides. It can be checked (see also appendix) that the constructed representation is one-sided (in particular, the 
 red polygons are contained in $\mathcal{P}$) and therefore,
 by~\autoref{lem:transform} it can be transformed into a (one-sided) side-contact representation.
We summarize our result.

\begin{theorem}\label{thm:lower}
	The graph $K_{3,8}$ has a side-contact representation 
with convex polygons in 3D.
\end{theorem}

\begin{figure}
\hspace*{-8mm}
	\hspace*{.7cm}\includegraphics[scale=.65,page=2]{k38-2.pdf}
	\hspace*{.7cm}\includegraphics[scale=.65]{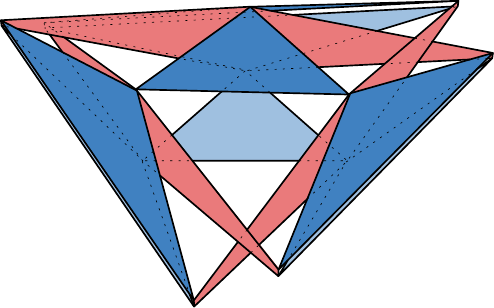}
	
	\caption{The configuration of the red polygons and the prism (left). The full configuration (right).}
	\label{abb:k38}
\end{figure}

We remark that the side-contact representation of $K_{3,8}$ is one-sided. As a consequence
of a result by Arseneva
et al.~\cite[Lemma~11]{akklsvw} no $K_{3,t}$ with $t>8$ has a one-sided representation with side-contacts.
	
\section{$K_{3,250}$ has no side-contact representation}\label{sec:ub}

In this section we prove the following result:
\begin{theorem}\label{thm:upper}
 The graph $K_{3,250}$ has no side-contact representation 
with convex polygons in 3D.
\end{theorem}

To prove \autoref{thm:upper} we present first some lemmas for
configurations of (straight-line) segments in 2D. 
Thus, until mentioned otherwise, all configurations
are considered in 2D from now on.
We say that a set of segments $\mathcal{S}$ is convex if every $s \in \mathcal{S}$ lies on the convex hull
of $\mathcal{S}$ and no two segments have the same slope.
We allow that in a convex set of segments two segments share an endpoint. Assume that the 
elements of
$\mathcal{S}$ are named such that the sequence $s_1,s_2\ldots,s_m$ lists the segments according to
their clockwise appearance on the convex hull. 
The intersection of the supporting lines of two segments $s_i$ and $s_j$
is called \emph{support intersection
point} (si-point for shorthand notation) of $s_i$ and $s_j$.
If $j=i+1$, or $i=1$ and $j=m$, we call the si-point of $s_i$ and $s_j$ a 
\emph{consecutive support intersection
point} (csi-point for shorthand notation).
Let $s_i=a_ib_i$ and $s_j=a_jb_j$ such that $a_ib_j$ is a proper segment
 on the convex hull of $\mathcal{S}$. The csi-point $c_{ij}$ of $s_i$ and $s_j$ is called
\emph{flopped} if the line through $a_ib_j$ defines a 
closed half-space that contains $c_{ij}$ and $\mathcal{S}$. See~\autoref{fig:crossingpoints} for an 
illustration.

\begin{figure}
	\centering
	\includegraphics[scale=.7]{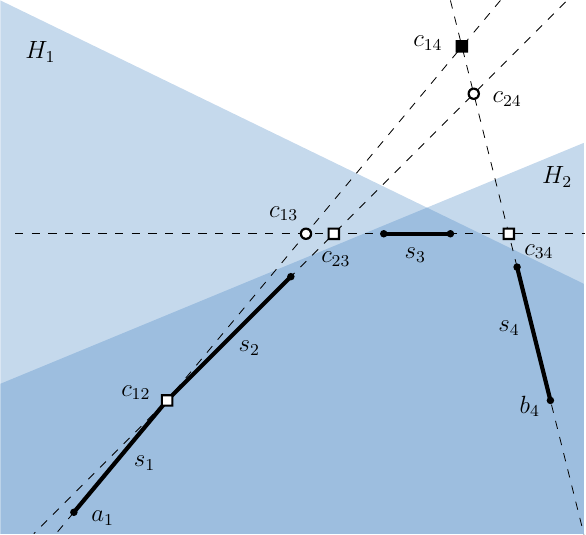}
	\caption{Four segments in convex position. Csi-points are shown as squares. The only flopped
	csi-point $c_{14}$ is filled. All other si-points are drawn as (empty) disks. 
	}
	\label{fig:crossingpoints}
\end{figure}

\begin{lemma}\label{lem:ub1}
	For any set $\mathcal{S}$ of segments in convex positions there is
	at most one csi-point that is flopped. \end{lemma}

\begin{proof}
	Let $s_i$ and $s_j$ be two segments with csi-point $c_{ij}$.
	We can assume that $j=i+1$.
	 The point $c_{ij}$ can only be flopped if the clockwise radial sweep of the 
	 tangent lines from $s_i$ to $s_j$ requires an angle larger than $\pi$, since 
	we need to transition a state in which the tangent line
	is parallel to $s_i$. Since in a total angular sweep we rotate by $2\pi$, this
	can happen only once. 
\end{proof}

\begin{lemma}\label{lem:ub2}
	Let $\mathcal{S}$ be a set of at least four segments in convex position.
	Consider any two closed half-spaces $H_1$ and $H_2$ that (i) both contain  $\mathcal{S}$, and (ii) no $s\in \mathcal{S}$ is completely part of the boundary of $H_1$ or $H_2$.
	 Then at least one csi-point of $S$
	lies in the interior of $H_1 \cap H_2$. 
\end{lemma}

\begin{proof}
	Assume first that $\mathcal{S}$ contains no flopped csi-point. Then
	the set of csi-points forms a convex set $C$. Furthermore, every edge of the
	convex hull of $C$ contains exactly one segment of $\mathcal{S}$ completely.
	Consider now a closed half-space $H$ that contains  $\mathcal{S}$. 
	If the interior of $H$
	misses two points from $C$, then an edge of the convex hull of $C$ (and therefore a segment of $\mathcal{S}$) lies in the complement of  
	the interior of $H$. Clearly $H$ violates condition (i) or (ii) from the statement of the lemma in this case.
		
	Now assume that we have a flopped csi-point $c$. We can augment $\mathcal{S}$ by adding a new
	segment such that the new set is convex and has no flopped csi-point. All 
	csi-points other than $c$ will remain. Thus, also in this situation, at most
	one nonflopped
	csi-point is not in the interior of $H$ 
	if the boundary of $H$ contains no segment from~$\mathcal{S}$ completely.
	
	The interior of the intersection of any two closed half-spaces $H_1$ and $H_2$ fulfilling~(i) and~(ii)
	can therefore miss no more than one nonflopped csi-point per half-space, and possibly 
	a flopped csi-point if it exists. By Lemma~\ref{lem:ub1} there can only be one flopped csi-point. 
	The statement of the lemma follows.
\end{proof}

We remark that the statement of~\autoref{lem:ub2} is \enquote{tight} as shown by the configuration in~\autoref{fig:crossingpoints}.

\begin{lemma}\label{lem:ub3}
Let $\mathcal{S}=\{s_1,\ldots,s_m\}$ be a set of segments in convex position 
indexed in cyclic  order. Assume that
 $s_1=aa'$ and $s_m=bb'$ define a flopped csi-point $c$ such that the segment $ab$ lies on the convex hull
of $\mathcal{S}$. Then all si-points of $\mathcal{S}$ lie inside the triangle $\Delta$ spanned
by $a$, $b$ and $c$.
\end{lemma}

\begin{proof}
Let $\ell_i$ be the supporting line of $s_i$. We denote the intersection of $\ell_i$ with $\ell_j$ by $x_{ij}$. Notice 
that on $\ell_1$  the following points appear in order: $a ,x_{12}, x_{13}, \ldots, x_{1m},c$.
On $\ell_m$ however, we have the order: $b, x_{m(m-1)},\ldots,x_{m2},x_{m1},c$. Both facts
can be observed by radially sweeping a tangent line around the convex hull of $\mathcal{S}$.

Now consider any two segments $s_i,s_j\in \mathcal{S}$, with $1<i<j<m$. Their si-point is denoted as
 $x_{ij}$. Since on $ac$ the order of points is $a,x_{1i},x_{1j},c$ and on $bc$ the order
of points is $b,x_{mj},x_{mi},c$ the segments $x_{1i}x_{mi}$ and $x_{1j}x_{mj}$ have to cross inside
$\Delta$ and hence the si-point defined by $s_i$ and $s_j$ lies in $\Delta$ (see also~\autoref{fig:ub3}). 
We have already observed that all other relevant si-points (defined by either $s_1$ or $s_m$) lie on the 
boundary of $\Delta$.

\begin{figure}
	\centering
	\includegraphics[scale=1]{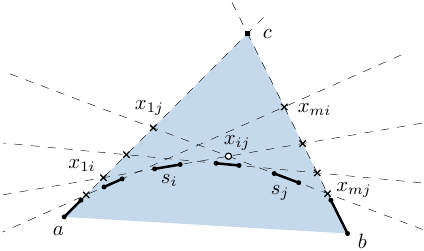}
	\caption{Illustration of the proof of~\autoref{lem:ub2}. The triangle $\Delta$ is shaded blue
	and the si-point induced by $s_i$ and $s_j$ is drawn as empty disk.}
	\label{fig:ub3}
\end{figure}
	
\end{proof}

We now prove \autoref{thm:upper} and go back to 3D.

\begin{proof}\textit{(\autoref{thm:upper}).}
Assume that we have a side-contact representation of $K_{3,250}$.
	 We call the polygons $r_1,r_2,r_3$ of the first partition class the \emph{red polygons}. 
	 The polygons of the other partition class are called the
\emph{blue polygons}. The supporting
plane of a polygon $r_i$ is named $r_i^=$.	 
	Let $\mathcal{A}$ be the arrangement given by  $r_1^=,r_2^=,r_3^=$. 
	We can assume that the three planes intersect in a single point $r_*$, and that no two sides of a polygon are parallel. Otherwise we apply suitable (small) projective transformations 
	to prevent parallel planes and lines without disconnecting the polygons. 
	 We call the eight (closed) cells of~$\mathcal{A}$ \emph{octants}. Note that every blue polygon
	has to lie in a single octant, since it has a side-contact with each of the red polygons.
	A red polygon can only be part of all octants if it contains $r_*$. 
	Thus, at least two red polygons need to avoid $r_*$ and \enquote{miss} at least
	two octants each, and only one of these
	octants can be the same. As a consequence there are at most 5 octants that have a piece
	of every red polygon on the boundary. One of them contains at least $50 = 250 /5$
	blue polygons. We denote this octant by $\mathcal{C}$.
	
	Let $f_i$ be the bounding face of $\mathcal{C}$ that contains $r_i$ and denote the interior of $f_i$
	by $\tilde f_i$. 
	Further let $\rho_{ij}$ be $f_i\cap f_j$. Both $r_1$ and $r_2$, can have at most one side fully contained in  $\rho_{12}$ and no side from $r_3$ can be completely in $\rho_{12}$ since $\rho_{12}\cap r^=_3=\{r_*\}$. Thus, $\rho:=\rho_{12}\cup\rho_{23}\cup\rho_{13}$
	contains at most six complete sides from blue polygons in $\mathcal{C}$. We ignore any blue polygon with
	a full side in $\rho$ and remain with a set $B$ of at least $44$ blue polygons. 
	
	First, we consider 
	the polygon $r_1$ and select a set~$\mathcal{S}_1$ of 44 of its sides that are incident to some polygon in $B$.
	 As usual, 
	we label the segments $
	s_1,s_2,\ldots, s_{44}$ cyclically
	and set $\mathcal{S}'_1=\{s_1,s_{12},s_{23},s_{34}\}$.
	The face $f_1$ can be obtained by intersecting
	$r_1^=$ with two closed half-spaces. 
	No segments of $\mathcal{S}'_1$ lies completely on the boundary of $f_1$ and thus, 
	by~\autoref{lem:ub2} at least
	one csi-point, say $c$, of $\mathcal{S}'_1$ lies in $\tilde f_1$.
	Take the two segments $aa'$ and $bb'$ (with $a'b'$ on the convex hull on $\mathcal{S}'_1$) 
	defining $c$ and all of the ten segments of $\mathcal{S}$ 
	in between them in the cyclic order. We call this set $\mathcal{S}''_1$. Note 
	that this set has a flopped csi-point, which is $c$. Since $a,b,c$ 
	lie in  $\tilde f_1$ we have by~\autoref{lem:ub3} that all si-points of $\mathcal{S}''_1$
	lie in $\tilde f_1$. 
	
	We now deal with polygon $r_2$. Let $\mathcal{S}_2$ be the set
	of sides of $r_2$ that share a side with a blue polygon that has a side in $\mathcal{S}''_1$.
	We get that $|\mathcal{S}_2|=12$.
	We sort the segments in $\mathcal{S}_2$ again by a radial sweep (notice that the order
	might be different than in $\mathcal{S}''_1$). This time we select the first, fourth,
	seventh and tenth segment in this order and we denote this subset by $\mathcal{S}'_2$.
	Again, we apply~\autoref{lem:ub2} to find a csi-point in $\tilde f_2$ and then~\autoref{lem:ub3}
	to obtain a set $\mathcal{S}''_2$ of (this time 4) segments, whose si-points are all in 
	 $\tilde f_2$.
	
	Finally, we consider $r_3$. Let $\mathcal{S}_3$ be the set
	of sides of $r_3$ that share a side with a blue polygon that has a side in $\mathcal{S}''_2$
	(and therefore in $\mathcal{S}''_1$ as well).
	Since $|\mathcal{S}_2|=4$ we get by~\autoref{lem:ub2} that one csi-point of $\mathcal{S}_3$
	lies in $\tilde f_3$. Two segments of $\mathcal{S}_3$ define this point. Call the adjacent blue polygons
	$b_1$ and $b_2$, with supporting planes $b_1^=$ and $b_2^=$. 
	We denote the restriction of $b_1^=/b_2^=$ to the boundary of $\mathcal{C}$ by $t_1/t_2$.
	By our construction,  
	$t_1$ and $t_2$ intersect on $\tilde f_3$ in a csi-point. 
	But both blue polygons have also a common side with each of the
	sets $\mathcal{S}''_2$ and $\mathcal{S}''_1$. As a consequence, $t_1$ and $t_2$ intersect
	in an si-point of $\mathcal{S}''_1$ inside $\tilde f_1$ and
	in an si-point of $\mathcal{S}''_2$ inside $\tilde f_2$. 
	The three si-points are distinct and define a plane. We get that $b_1^==b_2^=$.
	However, if two blue polygons lie in the same plane, all red polygons and therefore all 
	blue polygons have to lie in this plane as well. 
	Since $K_{3,250}$ is nonplanar, it has no side-contact
	representation in the plane, and we have obtained the desired contradiction.
\end{proof}

The following is now a simple consequence from the
K\H{o}vari--S\'os--Tur\'an
  theorem~\cite{KST54}, which states that an $n$-vertex graph that has
  no $K_{s,t}$ as a subgraph can have at most $O(n^{2-1/s})$
  edges.

\begin{corollary}
  \label{cor:density}
  Let $G$ be an $n$-vertex graph with a side-contact representation of convex polygons
  in 3D. Then the number of edges in $G$ is bounded by $O(n^{5/3})$.
  \end{corollary}

\bibliographystyle{plain}
\bibliography{contacts}

\appendix
\section{Omitted Details for the Construction in Section~\ref{sec:lb}}

We give here the full details how to obtain the one-sided
corner-contact representation of $K_{3,8}$. We start with the precise
definitions of the coordinates. Let $q_1,\ldots,q_8$ be the points of the
first red polygon $r_1$. For a point $q\in \mathbb{R}^3$ we denote its coordinates
by $(q^x,q^y,q^z)$. The coordinates for $r_1$ are chosen as  follows

\begin{center}
\begin{tabular}{>{\centering\arraybackslash}p{1cm}|*{8}{>{\centering\arraybackslash}p{1.2cm}}}
\toprule
$i$ & 1 & 2 & 3 & 4 & 5 & 6 & 7 & 8 \\
\midrule
$q_i^x$& $8$& $8$& $0.08$& $0$& $-8$& $-8$& $-0.08$ & $0$  \\
$q_i^y$& $2\sqrt{3}$ & $2\sqrt{3}$ &$2\sqrt{3}$ &$2\sqrt{3}$ &$2\sqrt{3}$ &$2\sqrt{3}$ &$2\sqrt{3}$ &$2\sqrt{3}$    \\
$q_i^z$& $-0.6$ & $-0.4$ & $1.1815$ & $1.1814$ & $0.6$ & $0.4$ & $-1.1815$ & $-1.1814$ \\
\bottomrule 	
\end{tabular}
\end{center}

\begin{figure}
	\centering
	\includegraphics[scale=1]{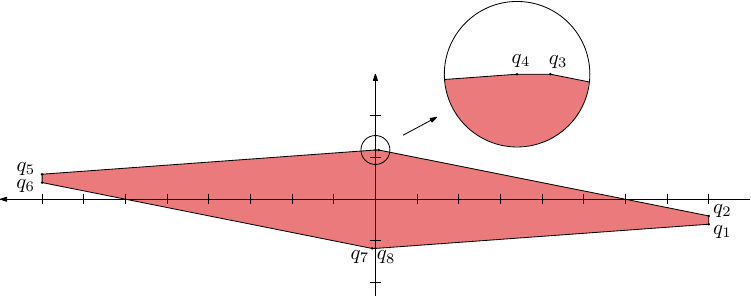}
	\caption{The shape of a red polygon.}
	\label{fig:red2d}
\end{figure}

The red polygon $r_2$ is obtained by rotating a copy $r_1$ around the z-axis by $2\pi/3$. Similarly, $r_3$ is
given by a copy of $r_1$ rotated by $-2\pi/3$ around the z-axis. We denote the vertices of $r_2$ by $s_1,\ldots,s_8$
and the vertices of $r_3$ by $t_1,\ldots,t_8$, such that $s_i/t_i$ are copies of $q_i$. We need to
assure that the three polygons are disjoint. When looking at the top view (\autoref{fig:top}) 
we see that there can be at most three possible intersections, marked with a cross in the figure. Due
to symmetry we only need to check one of these locations. Thus, it suffices to assure that segment $q_2q_3$ lies
below segment $t_6t_7$ at the line $v$ parallel to the z-axis through the point $x$, where $x=(6,2\sqrt{3},0)$ is the intersection
of $r_1$ and $r_3$ when projected into the xy-plane. Computing the distances shows that the 
z-coordinate of $q_2q_3\cap v$ is $-0.000631313$ and therefore negative. By symmetry $t_6t_7\cap v$
has a positive z-coordinate and hence the two red polygons are disjoint.

\begin{figure}
	\centering
	\includegraphics[scale=1]{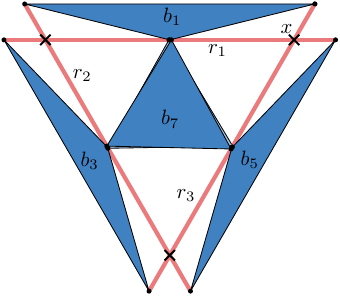}
	\caption{Top view of the representation.}
	\label{fig:top}
\end{figure}

To specify a blue polygon $b$,
we write $b=(q_i,s_j,t_k)$ if the corners of $b$ are $q_i,s_j,t_k$. The blue polygons are then given as follows:

\vskip1ex
\begin{centering}
\begin{tabular}{*{4}{>{\centering\arraybackslash}p{3cm}}}
$b_1=(q_4,s_5,t_2)$ & 	$b_2=(q_8,s_6,t_1)$ & $b_3=(q_2,s_4,t_5)$ & $b_4=(q_1,s_8,t_6)$ \\
$b_5=(q_5,s_2,t_4)$ & 	$b_6=(q_6,s_1,t_8)$ & $b_7=(q_3,s_3,t_3)$ & $b_7=(q_7,s_7,t_7)$ \\
\end{tabular}	
\end{centering}
\vskip1ex

It remains to check if the configuration of red and blue polygons is one-sided. To decide for a polygon $p$ on which side 
of the supporting plane $p^=$ a point $q$ lies we can use the signed volume of the tetrahedron spanned by
$q,q_1,q_2,q_3$, where $q_1,q_2,q_3$ are three arbitrary corners of $p$. This volume can be computed with the following expression:
$$
[q,q_1,q_2,q_3]:=\det \begin{pmatrix}
	q^x & q_1^x & q_2^x & q_3^x \\
	q^y & q_1^y & q_2^y & q_3^y \\
	q^z & q_1^z & q_2^z & q_3^z \\
	1 & 1 & 1 & 1 \\ 
\end{pmatrix}
$$
Thus, for every polygon $p'$ adjacent to  $p$ we have to check if the nonzero entries of $\{ [q',q_1,q_2,q_3] \mid q' \text{ is corner of $p'$}\}$ have all the same sign. We finish with a small python-3 script that will carry out the necessary computations. Running the script verifies
that the configuration is a one-sided corner-contact representation of $K_{3,8}$.

\begin{minted}{python}

import numpy as np
import math

def check_side(pa,polygonlist):
    """  Processes all polygons in polygonlist
    Prints "true" if all polygons in polygonlist are on
    one side of polygon pa. 
    """
    for pb in polygonlist:
        orientation = []
        for point in pb:
            m1 = np.array([pa[0],pa[1],pa[2],point])
            m2 = np.insert(m1,3,[1,1,1,1],1)
            orientation.append(np.linalg.det(m2))
        a = np.array(orientation)
        print(np.any((a <= 0)|(a >= 0 )))

 
red2d = np.array([[8,-.6],[8,-.4],[.08,1.1815],[0,1.1814],[-8,.6],
	[-8,.4],[-.08,-1.1815],[0,-1.1814]],)
red1 = np.insert(red2d,1,np.full((1,8),2*math.sqrt(3)),1)
rot_matrix = np.array([[-1/2,-math.sqrt(3)/2,0],
	[math.sqrt(3)/2,-1/2,0],[0,0,1]])
red2 = np.matmul(red1,rot_matrix)
red3 = np.matmul(red2,rot_matrix)

blue1 = np.array([red1[3],red2[4],red3[1]])
blue2 = np.array([red1[7],red2[5],red3[0]])
blue3 = np.array([red2[3],red3[4],red1[1]])
blue4 = np.array([red2[7],red3[5],red1[0]])
blue5 = np.array([red3[3],red1[4],red2[1]])
blue6 = np.array([red3[7],red1[5],red2[0]])
blue7 = np.array([red1[2],red2[2],red3[2]])
blue8 = np.array([red1[6],red2[6],red3[6]])

reds = [red1,red2,red3]
blues = [blue1,blue2,blue3,blue4,blue5,blue6,blue7,blue8]

for red in reds:
    check_side(red,blues)

for blue in blues:
    check_side(blue,reds)
	
\end{minted}
	
\end{document}